\documentclass[12pt]{article}
\usepackage{amsmath, amssymb, mathrsfs, bm, natbib}
\usepackage[usenames]{color}
\usepackage[dvips]{graphicx}
\usepackage{rotating}

\usepackage{setspace}
\usepackage{mathrsfs}
\usepackage{amssymb,amsmath,amsthm,mathrsfs,bm,verbatim}
\usepackage{indentfirst}
\usepackage{listings}
\usepackage{booktabs,multirow}

\newtheorem{proposition}{Proposition}

\numberwithin{equation}{section}
\begin{document}

\begin{center}
\Large{A Two-stage Online Monitoring Procedure for High-Dimensional Data Streams}
\end{center}

\begin{center}
{\large Jun Li}\\
     Department of Statistics, University of California - Riverside
\end{center}

\begin{abstract}

\noindent
Advanced computing and data acquisition technologies have made possible the collection of high-dimensional data streams in many fields. Efficient online monitoring tools which can correctly identify any abnormal data stream for such data are highly sought after. However, most of the existing monitoring procedures directly apply the false discover rate (FDR) controlling procedure to the data at each time point, and the FDR at each time point (the point-wise FDR) is either specified by users or determined by the in-control (IC) average run length (ARL). If the point-wise FDR is specified by users, the resulting procedure lacks control of the global FDR and keeps users in the dark in terms of the IC-ARL. If the point-wise FDR is determined by the IC-ARL, the resulting procedure does not give users the flexibility to choose the number of false alarms (Type-I errors) they can tolerate when identifying abnormal data streams, which often makes the procedure too conservative. To address those limitations, we propose a two-stage monitoring procedure that can control both the IC-ARL and Type-I errors at the levels specified by users. As a result, the proposed procedure allows users to choose not only how often they expect any false alarms when all data streams are IC, but also how many false alarms they can tolerate when identifying abnormal data streams. With this extra flexibility, our proposed two-stage monitoring procedure is shown in the simulation study and real data analysis to  outperform the exiting methods.

\end{abstract}


\section{Introduction}\label{sec:intro}

Recent advances in data acquisition technologies have greatly facilitated the collection of massive data in many industries. The number of data streams needed to be monitored in many statistical process control (SPC) applications has grown significantly. For example, in a network traffic surveillance application, L\'{e}vy-Leduc and Roueff (2009) considered the problem of monitoring thousands of Internet traffic metrics provided by a French Internet service provider. In a health-care monitoring application, Spiegelhalter et al. (2012) discussed how to simultaneously monitor 200,000 indicators of excess mortality in the UK health system. In the era of big data, the dimension of data streams available for monitoring will continue to grow. The demand for efficient monitoring schemes that can work for high-dimensional data streams has never been greater.

Different from monitoring a single data stream, a desired monitoring scheme for high dimensional data streams not only needs to raise an alarm as early as possible when some data streams are out-of-control (OC), but also needs to correctly identify those OC data streams. The research problem of monitoring high-dimensional data streams for this purpose has attracted considerable attention in the literature recently. A number of monitoring schemes have been proposed, including those in Grigg and Spiegelhalter (2008), Li and Tsung (2009, 2012), and Gandy and Lau (2013). A common theme in those methods is to monitor each data stream by one control chart, and calculate the $p$-values of the charting statistics, and then apply some existing false discover rate (FDR) controlling procedure from the multiple hypothesis testing literature to those $p$-values at each time point. Since the FDR controlling procedure is used separately at each time point in those methods, the FDR is only controlled locally at each time point. We refer to the FDR at each time point as the point-wise FDR hereafter. However, when designing a monitoring scheme, it is more desirable to have control of certain global measures, for example, in this case, control of a global FDR. Because it has long been believed to be true  that if the point-wise FDR is controlled locally at each time point, the FDR over any window of time (the global FDR) will be controlled at the same level (see, for example, Grigg and Spiegelhalter (2008)), the approach based on controlling the point-wise FDR has dominated the literature for developing monitoring schemes for high-dimensional data streams. In Section 2, we show that, contrary to the common perceptions,  control of the point-wise FDR does not guarantee control of the global FDR.

In addition to the lack of global FDR control, the above point-wise FDR monitoring schemes do not directly tell users what is the expected in-control (IC) average run length (ARL). The IC-ARL is the expected number of observations collected before any false alarm occurs in any data stream when the system is IC. It is widely used to measure how often to expect a false alarm when the system is IC in the literature. In many applications, the system is usually expected to be IC most of the time, and users do not want a monitoring scheme to raise any false alarm too often. Therefore, the IC-ARL is also a reasonable global measure to consider when designing monitoring schemes for high-dimensional data streams.  Different from the other methods mentioned above, which simply apply the FDR controlling procedure at each time point using the point-wise FDR pre-specified by users, Li and Tsung (2009, 2012) designed their monitoring scheme to have the desired IC-ARL. More specifically, they determined the point-wise FDR used in the FDR controlling procedure at each time point through Monte-Carlo simulation so that the resulting monitoring scheme has the desired IC-ARL.  Although their procedure has control of a global measure, the IC-ARL, the point-wise FDR used in their FDR controlling procedure is completely determined by the specified IC-ARL, and a larger IC-ARL implies a smaller point-wise FDR. However, the point-wise FDR reflects how much users can tolerate false alarms (Type-I errors) when identifying OC data streams. The smaller the point-wise FDR, the less powerful the monitoring scheme for detecting OC data streams.  Therefore, the monitoring scheme in Li and Tsung (2009, 2012) can not have both the IC-ARL and Type-I errors controlled at the user's desired level. One has to be compromised.

From the above discussion, we can see that an attractive monitoring scheme for high-dimensional data streams should be able to achieve global control of certain false alarm rate and at the same time allow users to choose the level of Type-I errors they can tolerate when identifying OC data streams. As we show later, the existing monitoring procedures mentioned above were developed using a single-stage strategy. The nature of this single-stage strategy makes it difficult to modify those procedures in order for them to have the above two desired properties. In this paper, different from the single-stage strategy used in those existing monitoring schemes, we propose a general two-stage strategy to design high-dimensional data monitoring schemes. As shown in Section 3, the high-dimensional data monitoring problem can be formulated in a way that we essentially need to answer the following two questions at each time point when monitoring high-dimensional data: the first question is if there are any OC data streams, and if the answer to this question is ``yes'',  then the second question is where the OC data streams are. Based on this formulation, we propose the following two-stage procedure. In the first stage, at each time point a global test is conducted based on information gathered across all the data streams to decide if there are any OC data streams. This answers the ``IF'' question. The decision rule for this global test can be naturally used to satisfy the global IC-ARL requirement. If it is determined that there exists at least one OC data stream from the first stage, we move to the second stage and carry out some local tests to identify which data streams are OC. This answers the ``WHERE'' question. The decision rule for those local tests will be determined to control certain Type-I error rates. As a result, our proposed two-stage procedure allows users to choose not only how often they expect a false alarm when the system is IC (the IC-ARL requirement) but also how many false alarms they can tolerate  when identifying OC data streams (the Type-I error rate requirement). Due to this extra flexibility, our two-stage procedure can significantly outperform the method proposed in Li and Tsung (2009, 2012) as shown in our simulation studies.

The rest of the paper is organized as follows. In Section 2, we provide some useful insight of the existing point-wise FDR monitoring schemes and explain why those monitoring schemes lack control of the global FDR. In Section 3, we introduce our general two-stage procedure for monitoring high-dimensional data streams. In Section 4,  we demonstrate the use of this general approach in developing a new monitoring scheme for detecting Gaussian data streams with location shifts. A simulation study is reported in Section 5 to compare the performance of this new monitoring scheme with some existing method. In Section 6, we demonstrate the application of our proposed two-stage procedure using a real data set from a manufacturing process. Finally, we provide some concluding remarks in Section 7.

\section{Results on existing high-dimensional data monitoring procedures}
The setup for our high-dimension data monitoring problem is the following. There are $m$ data streams in the system, where $m$ can be hundreds or thousands.  We denote the observation from the $i$-th data stream at time $t$ by $X_{i,t}$, $i=1,...,m$, $t=1,2,...$. Since a time series model can be used to decorrelate the temporal correlation between the observations from each data stream and a spatial model can be used to decorrelate the spatial correlation between the data streams before applying monitoring schemes,
without loss of generality, we assume that the $X_{i,t}$ are independent both within each data stream and between different data streams. When the system is in-control, the underlying distribution of $\{X_{i,1},X_{i,2}, ...\}$ is called the in-control distribution, denoted by $F_{i,0}$, $i = 1,...,m$.

Following the above setup, at any time $t$, we observe $X_{i,1}, X_{i,2}, ..., X_{i,t}$, $i=1,...,m$, and  the task of our online monitoring scheme at time $t$ is to determine if the distribution of $X_{i,1}, X_{i,2}, ..., X_{i,t}$ is the same as $F_{i,0}$ for all $i=1,2, ...,m$. When some of the distributions of $X_{i,1}, X_{i,2}, ....,X_{i,t}$ have changed from $F_{i,0}$, we not only want to trigger an alarm as early as possible, but also want to correctly identify those data streams that have changed from $F_{i,0}$. This is equivalent to conducting the following multiple hypothesis testing, for $i=1,...,m$,
%In other words, we want to decide whether there exists $1 %\leq \tau \leq t$ %such that  $X_1, \cdots, %X_{\tau-1}$ %follow distribution $F_0$, while %$X_{\tau}, \cdots, %X_t$ follow some %distribution $F^*\neq F_0$.  %$\tau$ is usually called %change-point. It amounts to %testing
\[
H_{0,i,t}: X_{i,1}, \cdots, X_{i,t} \text{ follows } F_{i,0},
\]
versus
\begin{align}
H_{1,i,t}: \exists \text{ } \tau_i \in [1,t] \text{ such that } & X_{i,1}, \cdots,  X_{i,\tau_i-1} \text{ follows } F_{i,0} \nonumber \\
\text{ and } & X_{i,\tau_i}, \cdots, X_{i,t} \text{ follows } F_{i,1},
\label{test0}
\end{align}
where $F_{i,1}\neq F_{i,0}$, $F_{i,1}$ is  the OC distribution and $\tau_i$ is the change-point for the $i$-th data stream.

Most of the existing methods use the above formulation to design their monitoring schemes, which we refer to as ``the single-stage strategy''. Denote the $p$-value based on some appropriate test statistic for the above hypothesis testing problem by $p_{i,t}$. Then most of the existing monitoring schemes are constructed by applying some existing FDR controlling procedure, such as the step-up procedures proposed by Benjamini and Hochberg (1995) and Benjamini and Yekutieli (2001), to those $p_{i,t}$ at each time point with some FDR specified by users. The scheme sets off an alarm if the FDR controlling procedure produces at least one rejection and the data stream corresponding to the rejected null hypothesis is then identified as the OC data stream. In the following, we show that control of the point-wise FDR in those methods does not guarantee control of the global FDR.

\begin{table}[!htbp]
\begin{center}
\caption{The outcome of the multiple hypothesis testing problem}
\label{tab:FDR}
\begin{tabular}{lccc}
  \hline
 & \textbf{Number} &  \textbf{Number}  & \\
 \textbf{Number of} & \textbf{not rejected} & \textbf{rejected} & \\
 \hline
 IC data stream (True null hypothesis) & $m_{0,t}-V_t$ & $V_t$ & $m_{0,t}$\\
 OC data stream (Non-true null hypothesis) & $m_{1,t}-S_t$ & $S_t$ & $m_{1,t}$ \\
 \cline{1-4}
 & $m-R_t$ & $R_t$ & $m$\\
  \hline
\end{tabular}
\end{center}
\end{table}

To begin with, we introduce a few notations. Denote the unknown numbers of IC and OC data streams at time $t$ by $m_{0,t}$ and $m_{1,t}=m-m_{0,t}$, respectively. In other words, $m_{0,t}$ and $m_{1,t}$ are the unknown numbers of true and false null hypotheses for the hypothesis testing problem in (\ref{test0}), respectively. Let $R_t$ be the number of rejected hypotheses based on some rejection rule. Among those $R_t$ rejected hypotheses, $V_t$ (unknown) of them are falsely rejected and $S_t$ (unknown) of them are correctly rejected. Those $V_t$ rejections are false alarms or Type I errors. The situation discussed above is summarized in Table~\ref{tab:FDR}. The point-wise FDR at time $t$ is then defined as the expected proportion of Type-I errors among the rejected hypotheses at time $t$, that is, $\text{LFDR}_t=E(Q_t)$, where $Q_t=V_t/R_t$ if $R_t>0$ and 0 if $R_t=0$. Alternatively, it can be written as
\[
\text{LFDR}_t=E\Big(\frac{V_t}{R_t \vee 1} \Big).
\]
Similarly, an appropriate way to define the global FDR up to time $T$ is
\[
\text{GFDR}_T=E\Big(\frac{\sum_{t=1}^T V_t}{ (\sum_{t=1}^T R_t) \vee 1}\Big).
\]

When the system is IC, that is, all the null hypotheses are true with $m_{0,t}=m$ for $t=1,...,T$, the global FDR and point-wise FDR are then equivalent to,
\[
\text{GFDR}_T=P(\sum_{t=1}^T R_t \geq 1) \quad \text{ and } \quad \text{LFDR}_t=P(R_t \geq 1).
\]
Since most of the point-wise FDR monitoring procedures can have the point-wise FDR controlled at the level of $\alpha$, then $P(R_t \geq 1)$ is controlled at the level of $\alpha$ for every $t$. But it does not guarantee that $P(\sum_{t=1}^T R_t \geq 1)$ is also controlled at the level of $\alpha$. It is easy to see that, if $T$ is large enough, we will eventually encounter a false alarm. In other words,  $P(\sum_{t=1}^T R_t \geq 1)$ will approach 1. Therefore, the global FDR can approach 1 in this case. When the system is OC, as shown in our simulation studies reported in Section 5, the global FDR can be also much higher than the point-wise FDR.

To see it theoretically, note that
\[
\text{GFDR}_T=E\Big(\frac{\sum_{t=1}^T V_t}{ (\sum_{t=1}^T R_t) \vee 1}\Big) \leq E\Big(\max_{t=1,...,T}\Big\{\frac{V_t}{R_t \vee 1}\Big\}\Big) .
\]
Since
\[
E\Big(\max_{t=1,...,T}\Big\{\frac{V_t}{R_t \vee 1}\Big\}\Big) \neq \max_{t=1,...,T}E\Big(\frac{V_t}{R_t \vee 1}\Big)=\max_{t=1,...,T}\text{LFDR}_t,
\]
this helps explain why control of the point-wise FDR does not guarantee control of the global FDR in those point-wise FDR monitoring procedures.

To develop a procedure which can have the global FDR control, we notice that
\begin{align*}
&E\Big(\max_{t=1,...,T}\Big\{\frac{V_t}{R_t \vee 1}\Big\}\Big)=\int_0^{\infty}P\Big(\max_{t=1,...,T}\Big\{\frac{V_t}{R_t \vee 1}\Big\}\geq x\Big)\, dx\\
=& \int_0^{\infty}P\Big(\bigcup_{t=1,...,T}\Big\{\frac{V_t}{R_t \vee 1}\geq x \Big\}\Big)\, dx \\
\leq & \int_0^{\infty}\sum_{t=1}^TP\Big(\frac{V_t}{R_t \vee 1}\geq x  \Big)\, dx  \quad \quad \quad \quad (\text{by the Bonferroni inequality}) \\
= & \sum_{t=1}^T E\Big(\frac{V_t}{R_t \vee 1}\Big).
\end{align*}
Therefore, we have $\text{GFDR}_T \leq \sum_{t=1}^T \text{LFDR}_t$.  This implies that, if we can control the point-wise FDR at the level of $\alpha/T$, the global FDR is controlled at the level of $\alpha$. However, in many applications, $T$ is usually large, and then the level for the point-wise FDR, $\alpha/T$, can be extremely small, which will make the resulting procedure extremely conservative for detecting OC data streams. Therefore, although the above Bonferroni-type of adjustment gives us a simple procedure which has the desired global FDR control, the procedure might not be very attractive for use in practice. How to design a monitoring scheme which has proper control of global FDR remains an open problem.

Even if we were able to design a monitoring scheme that has proper control of global FDR, similar to the monitoring scheme in Li and Tsung (2009, 2012), we would not have the flexibility of choosing the point-wise FDR at the user's desired level. This is simply the result of using the singe-stage strategy. As mentioned in the Introduction, it is quite  challenging to develop a monitoring scheme based on the single-stage strategy that can have the control of certain global false alarm rate and at the same time allow the users to choose the level of Type-I errors they can tolerate when identifying OC data streams. In the following, we propose a monitoring scheme using a two-stage strategy that can easily have the two desired properties.

\section{Proposed two-stage monitoring procedure}
As shown in the previous section,  the online high-dimensional data monitoring problem can be formulated into a hypothesis testing problem in (\ref{test0}) at each time point $t$. However, this single-stage strategy leads to many restrictions of the resulting procedures. To help alleviate those restrictions, we formulate the monitoring problem into a two-stage hypothesis testing problem. At the first stage, we conduct the following global test by pulling all the information across the $m$ data streams,
\[
H_{0,t}: X_{i,1}, \cdots, X_{i,t} \text{ follows } F_{i,0}, \text{ for } i=1,...,m,
\]
versus
\begin{align}
\label{test1}
H_{1,t}: \exists \, \text{ at least one } i \in \{1,2,...,m\} \text{ and } \tau_i \in [1,t] \text{ such that } & X_{i,1}, \cdots, X_{i,\tau_i-1} \text{ follows } F_{i,0} \nonumber\\
\text{ and } & X_{i,\tau_i}, \cdots, X_{i,t} \text{ follows } F_{i,1},
\end{align}
If $H_{0,t}$ is accepted, it indicates that there is no OC data stream and all data streams are IC at time $t$. A rejection of $H_{0,t}$ indicates that some of the data streams are OC, in which case, we move to our second stage and conduct the multiple hypothesis testing in (\ref{test0}) to identify which data streams are OC under the condition that $H_{0,t}$ is rejected in the first stage.

In the above two-stage formulation, we can have the control limit developed for the hypothesis testing problem in each of the two stages, which allows us to design our monitoring scheme to satisfy both the IC-ARL and Type-I error rate requirements. In the following, we describe our two-stage monitoring procedure in detail.

\subsection{First stage: Controlling the IC-ARL}

In the first stage, we test the hypotheses in (\ref{test1}) at each time $t$. Denote an appropriate test statistic for this hypothesis testing problem by $G_t$, which is usually a function of the observations from all data streams up to time $t$. Without loss of generality, we assume that $H_{0,t}$ is rejected if $G_t$ is too large. Then we can use $G_t$ as our global monitoring statistic in the first stage. If $G_t$ is smaller than some control limit, denoted by $h$, we declare that there is no OC data stream and all data streams are IC. If $G_t$ is larger than $h$, we signal an alarm, suggesting some of the data streams are OC. Based on this formulation, the monitoring in the first stage is similar to the standard monitoring of one data stream, and we can choose the control limit $h$ such that the IC-ARL of the above monitoring scheme is equal to some pre-specified value. \textit{Therefore, by choosing the control limit $h$ in this stage, users are able to choose how often they want a false alarm to occur when the system is IC.}

Although we can choose any test statistic for the hypotheses in (\ref{test1}) as our global monitoring statistic $G_t$ and the control limit $h$ can be chosen accordingly to satisfy the IC-ARL requirement, we need to choose $G_t$ carefully so that it can be powerful to reject $H_{0,t}$ for different OC scenarios. How to develop an efficient global monitoring statistic has been also an active research area in SPC. In Section 4, we will give some examples of such $G_t$.

\subsection{Second stage: Controlling the Type-I error rate}
\subsubsection{Decision rule}
In the first stage we monitor the global monitoring statistic $G_t$. If $G_t<h$, we declare that there is no OC data stream and continue to monitor at the next time point. If $G_t>h$, we signal an alarm, suggesting some of the data streams are OC, and move to the second stage of our monitoring procedure, which is to identify which data stream is OC. As mentioned earlier, this task is equivalent to the multiple hypothesis testing problem in (\ref{test0}) under the condition that $G_t>h$. Denote the standardized test statistic we use for the hypothesis testing problem in (\ref{test0}) by $W_{i,t}$ and assume that $H_{0,i,t}$ is rejected if $W_{i,t}$ is too large. To decide which data stream is OC, a natural decision rule is the following: if $W_{i,t}$ is larger than some control limit, say $c_h$, we reject $H_{0,i,t}$ and conclude that the $i$-th data stream is OC; otherwise, we do not reject $H_{0,i,t}$ and conclude that the $i$-th data stream is IC. Here the subscript ``h'' in the control limit $c$ is to show that it depends on the control limit $h$ we obtain from the first stage.  Based on the above decision rule, our remaining task is to determine the threshold $c_h$. Since this is a multiple hypothesis testing problem, we need to determine $c_h$ to control certain Type-I error rate. The Type-I error rate in our monitoring application simply reflects how much we can tolerate false alarms when identifying abnormal data streams. \textit{Therefore, by choosing the control limit $c_h$ in our second stage, users are able to choose how many false alarms they want to tolerate when identifying OC data streams.}

\subsubsection{Choice of Type-I error rate}
When testing a single hypothesis, the Type-I error rate is simply the probability of a Type-I error. When testing multiple hypotheses in (\ref{test0}), there is a Type-I error associated with each $H_{0,i,t}$, $i=1,...,m$. Therefore, there are many ways to define the overall Type-I error rate when testing those multiple hypotheses simultaneously. Besides the FDR, some of the standard Type-I error rates used in the multiple hypothesis testing literature are the following (see, for example, Dudoit et al. (2003)):
\begin{itemize}
\item[$\bullet$] \textit{The per-comparison error rate} (PCER) is defined as the expected number of Type-I errors divided by the number of hypotheses.
\item[$\bullet$] \textit{The family-wise error rate} (FWER) is defined as the probability of at least one Type-I error.
\end{itemize}
Using the notations introduced in Table \ref{tab:FDR},
\[
\text{PCER}_t=E(V_t)/m \quad \text{ and } \quad \text{FWER}_t=P(V_t \geq 1).
\]

The FWER control is the most conservative of the above three Type-I error rates, and is usually considered to be too conservative for the high-dimensional data monitoring purpose. The FDR control is less conservative than the FWER control, and it is a popular choice in many multiple hypothesis testing problems. However, similar to what we explain in Section 2, it is difficult to figure out a way to determine $c_h$ such that the resulting procedure can have proper control of global FDR.
If we use the PCER as the choice of Type-I error rate, we can define the point-wise PCER and global PCER as follows,
\[
\text{LPCER}_t=E(V_t)/m \quad \text{ and } \quad \text{GPCER}_T=E(\sum_{t=1}^T V_t)/(Tm).
\]
Notice that the global PCER is simply the average of the point-wise PCERs, which implies that, if all the point-wise PCERs are controlled at the level of $\alpha$, the global PCER is controlled at the level of $\alpha$. Therefore, we can focus on the point-wise PCER when developing the decision rule. This is one of the advantages of using the PCER over the FDR. Furthermore, by definition, $m\text{LPCER}_t$ is simply the expected number of Type-I errors (false alarms). As a result, it is straightforward for users to specify the desired level of $\text{LPCER}_t$ from the number of false alarms they can tolerate when identifying OC data streams.  Due to the above reasons, in the following we use  the PCER as the choice of Type-I error rate in our second stage and describe how to determine $c_h$ to control the PCER.

\subsubsection{Determining the control limit based on the PCER}

Without loss of generality, we assume that the first $m_{0,t}$ data streams are IC. Then for any given $c_h$, we have
\[
\text{PCER}_t=E(V_t)/m=\frac{1}{m}\sum_{i=1}^{m_{0,t}} P_{H_{1,t}}\Big(W_{i,t} > c_h|G_t>h \Big)
\]
where the subscript ``$H_{1,t}$'' indicates that the probabilities are calculated under the setting where the first $m_{0,t}$ data streams are IC and the remaining are OC.  Therefore, for the PCER at the pre-specified level $\alpha$, we need to determine $c_h$ such that
\begin{equation}
\label{eqn:PCER}
\frac{1}{m}\sum_{i=1}^{m_{0,t}} P_{H_{1,t}}\Big(W_{i,t} > c_h|G_t>h \Big) \leq \alpha.
\end{equation}
The above conditional probabilities are very challenging to calculate. To circumvent this difficulty, we denote the marginal $p$-value of $W_{i,t}$ by $p_{i,t}$ and the probability density function of $p_{i,t}$ when the data stream is OC by $\ell_{i,t}(\cdot)$. As shown in the Appendix, if $\ell_{1,t}(\cdot),..., \ell_{m,t}(\cdot)$ are all non-increasing functions, we have, for $i=1,...,m_{0,t}$,
\[
P_{H_{1,t}}\Big(W_{i,t} > c_h|G_t>h \Big) \leq P_{H_{0,t}}\Big(W_{i,t} > c_h|G_t>h \Big),
\]
where the subscript ``$H_{0,t}$'' indicates that the probabilities are calculated under the null hypothesis $H_{0,t}$ in (\ref{test1}); that is, $m_{0,t}=m$ and all  data streams are IC.
Therefore,
\[
   \frac{1}{m}\sum_{i=1}^{m_{0,t}} P_{H_{1,t}}\Big(W_{i,t} > c_h|G_t>h \Big) \leq \frac{1}{m}\sum_{i=1}^{m}P_{H_{0,t}}\Big(W_{i,t} > c_h|G_t>h \Big).
\]
To find $c_h$ that satisfies (\ref{eqn:PCER}), we only need to find $c_h$ such that
\begin{equation}
\label{eqn:condP}
\frac{1}{m}\sum_{i=1}^{m}P_{H_{0,t}}\Big(W_{i,t} > c_h|G_t>h \Big)=\alpha.
\end{equation}

Note that the above condition that $\ell_{1,t}(\cdot),..., \ell_{m,t}(\cdot)$ are all non-increasing functions is not restrictive and is satisfied for most of the commonly used charting statistics, including the CUSUM statistic we used in Section 4.
%\[
%\frac{1}{m}\sum_{i=1}^{m}P_{H_{0,i,t}}\Big(W_{i,t} > c_h|G_t>h \Big) \leq %\alpha.
%\]
%The above conditional probability $P_{H_{0,i,t}}\Big(W_{i,t} > c_h|G_t>h \Big)$ %is very challenging to calculate, since we only know that the $i$-th data %stream is IC and do not know whether other data streams are IC or OC.

%To circumvent this difficulty,
%we first prove the following result.
%\begin{proposition}
%\label{thm:inequality}
%If $G_t$ is a non-decreasing function of $W_{i,t}$, then we have,
%we note that
%\[
%P_{H_{0,i,t}}\Big(W_{i,t} > c_h|G_t>h \Big) \leq P_{H_{0,t}}\Big(W_{i,t}> %c_h|G_t>h \Big),
%\]
%where the subscript ``$H_{0,t}$'' indicates that the probabilities are %calculated under the null hypothesis $H_{0,t}$ in (\ref{test1}); that is, all  %data streams are IC.
%\end{proposition}
%Therefore, if we can find $c_h$ such that
%\begin{equation}
%\label{eqn:condP}
%\frac{1}{m}\sum_{i=1}^{m}P_{H_{0,t}}\Big(W_{i,t} > c_h|G_t>h \Big) \leq %\alpha,
%\end{equation}
%we can have the PCER controlled at the level of $\alpha$.

To make the above task for finding $c_h$ easier, we assume that both $G_t$ and $W_{i,t}$ ($i=1,...,m$) follow their respective null steady-state distributions when $H_{0,t}$ is true. Under this assumption, $P_{H_{0,t}}\Big(W_{i,t} > c_h|G_t>h \Big)$ does not depend on $t$. Then given the control limit $h$ we obtain from the first stage, the control limit $c_h$ of our second stage can be found as follows.
\begin{itemize}
\item[\textbf{Step 1:}] For any given $c_h$, generate $m$ IC data streams and keep monitoring them using the global monitoring statistic $G_t$. When $G_t$ is larger than the given control limit $h$, record the proportion of $W_{i,t}$ that is larger than the given value of $c_h$ and denote it by $\hat{p}$.
\item[\textbf{Step 2:}] Repeat Step 1 $B$ times (for example, $B=2500$). Then the average of the $\hat{p}$ over the $B$ replications can be used to approximate $\frac{1}{m}\sum_{i=1}^{m}P_{H_{0,t}}\Big(W_{i,t} > c_h|G_t>h \Big)$.
\item[\textbf{Step 3:}] Using Steps 1 and 2, we can obtain the approximate of $\frac{1}{m}\sum_{i=1}^{m}P_{H_{0,t}}\Big(W_{i,t} > c_h|G_t>h \Big)$ for any given $c_h$. We can then use some numerical method, for example, the bisection method to obtain the control limit $c_h$ of our second stage that satisfies (\ref{eqn:condP}).
\end{itemize}

In the following, we summarize our two-stage monitoring procedure.
\begin{itemize}
\item[\textbf{Stage 1}:] We first select an appropriate global monitoring statistic $G_t$. The monitoring scheme in the first stage is simply calculating $G_t$ at each time point and comparing $G_t$ with the pre-specified control limit $h$. If $G_t$ is larger than $h$, we signal an alarm indicating that there exists at least one OC data stream, and proceed to the second stage. Otherwise, we continue the monitoring using $G_t$. The control limit $h$ is determined by using Monte Carlo simulation and bisection method such that the IC-ARL of the above monitoring scheme based on $G_t$ is equal to the pre-specified IC-ARL.

\item[\textbf{Stage 2}:] If $G_t$ is larger than $h$ and an alarm is given in the first stage, we then construct the local monitoring statistic $W_{i,t}$ for each data stream. If $W_{i,t}$ is larger than the control limit $c_h$, the $i$-th data stream is OC. Otherwise, it is IC. The control limit $c_h$ in this stage is determined using the three steps described above.
\end{itemize}

From the summary of our two-stage procedure, the control limit $h$ in our first stage is determined by the IC-ARL, which controls how often we expect a false alarm when the system is actually IC. The control limit $c_h$ in our second stage is determined by the PCER, which controls how many false alarms we can tolerate. Therefore, by using our two-stage monitoring procedure, the users can design their own monitoring scheme depending on how often they want to expect any false alarm when the system is IC and how many false alarms they can tolerate when identifying OC data streams.

\section{Two-stage monitoring procedure for Gaussian data streams with location shifts}
In this section, we show how to develop a two-stage monitoring procedure for high-dimensional data streams in one particular setting following the general approach we describe in Section 3. The setting we consider here is that the IC distribution $F_{i,0}$ ($i=1,...,m$) is the Gaussian distribution with mean $\mu_{i,0}$ and variance 1 (denoted by $N(\mu_{i,0},1)$), and the target OC distribution $F_{i,1}$ ($i=1,...,m$) is also some Gaussian distribution with mean $\mu_{i,1}$  and variance $1$ (denoted by $N(\mu_{i,1},1)$). Under the Gaussian distribution assumption, an optimal monitoring statistic for each data stream is the CUSUM statistic.  For example, to detect an upward mean shift, the CUSUM statistic for the $i$-th data stream is defined as
\begin{equation}
\label{eqn:CUSUM}
\begin{cases}
C_{i,0}^+=0\\
C_{i,t}^+=\max(0,C_{i,t-1}^++X_{i,t}-\mu_{i,0}-k_i)), \text{ for } t \geq 1
\end{cases}
\end{equation}
where $k_i=(\mu_{i,1}-\mu_{i,0})/2$ is the reference value and usually specified before the monitoring based on the target OC distribution mean $\mu_{i,1}$.

To implement our two-stage monitoring procedure, we first need to find an appropriate global monitoring statistic $G_t$, which has good power to detect any type of OC scenarios. Tartakovsky et al. (2006) proposed using  $\max_{i=1,...,m}C^+_{i,t}$ as the global monitoring statistic. This approach enjoys some optimal property when exactly 1 out of $m$ data streams is OC. However, it performs poorly if many data streams are OC. To overcome this limitation, Mei (2010) proposed using $\sum_{i=1}^m C^+_{i,t}$ to monitor the $m$ data streams simultaneously. It has been shown that  $\sum_{i=1}^m C^+_{i,t}$ is more effective than $\max_{i=1,...,m}C^+_{i,t}$ when a moderate or large number of data streams are OC. However,  $\sum_{i=1}^m C^+_{i,t}$ is less powerful than $\max_{i=1,...,m}C^+_{i,t}$ when only a few data streams are OC. In practice, it is usually unknown in advance how many data streams will be OC. Therefore, neither  $\max_{i=1,...,m}C^+_{i,t}$ nor $\sum_{i=1}^m C^+_{i,t}$ as the global monitoring statistic can guarantee robust performance. Xie and Siegmund (2013) recognized those limitations and further proposed a monitoring statistic derived from the likelihood function of a normal-mixture model. Besides being extremely computationally intensive,  their approach also needs to pre-specify the mixing percentage. The approach will perform the best if the pre-specified mixing percentage correctly reflects the percentage of OC data streams. If the mixing percentage is misspecified, their approach will  sacrifice some power. Recently Zou et al. (2014) proposed some alternative statistic to combine the CUSUM statistics from each of the $m$ data streams to produce a single global monitoring statistic. This approach does not need the prior knowledge of how many data streams are OC, and performs well comparing with the aforementioned methods across different OC scenarios. Therefore, in the following we use the global monitoring statistic proposed by Zou et al. (2014) as an example to demonstrate our two-stage procedure.

Denote the marginal $p$-value of $C^+_{i,t}$ defined in (\ref{eqn:CUSUM}) by $p_{i,t}$, $i=1,...,m$. The corresponding order statistics are $p_{(1),t}\leq \cdots \leq p_{(m),t}$. The global monitoring statistic $G_t$ proposed in Zou et al. (2014) is defined as,
\begin{equation}
\label{eqn:Gt}
G_t=\sum_{i=1}^m\Big\{\log\Big[\frac{(1-p_{(i),t})^{-1}-1}{(m-1/2)/(i-3/4)-1} \Big] \Big\}^2I_{\{p_{(i),t}<1-(i-3/4)/m\}},
\end{equation}
where $I_{\{A\}}$ is the indicator function and takes 1 if $A$ is true and 0 otherwise.

To use the above $G_t$ in the monitoring scheme, it is important that those $p$-values, $p_{i,t}$, can be calculated quickly. Grigg and Spiegelhalter (2008) developed an empirical approximation to the null steady-state distribution of the CUSUM statistic defined in (\ref{eqn:CUSUM}) and further provides a close-form formula to approximate the steady-state $p$-value. Since this formula only works when the CUSUM statistic reaches its steady-state, to make use of this formula, we modify the definition of the CUSUM statistic a little. Instead of starting the CUSUM statistic at 0, i.e., $C_{i,0}^+=0$ as in (\ref{eqn:CUSUM}), we start the CUSUM statistic at some value randomly drawn from the null steady-state distribution of $C_{i,t}^+$. To draw a random sample from the null steady-state distribution of $C_{i,t}^+$, we generate $10^6$ independent sequences of $\{X_{j,1},...,X_{j,2000}\}$ ($j=1,...,10^6$), each of which is independently drawn from $N(0,1)$, and calculate $C_{j,2000}^+$ as in (\ref{eqn:CUSUM}). Then $\{C_{j,2000}^+\}_{j=1}^{10^6}$ can serve as a random sample from the null steady-state distribution of $C_{i,t}^+$. Our modified CUSUM statistic is  then  defined as follows. For $i=1,...,m$,
\begin{equation}
\label{eqn:CUSUM1}
\begin{cases}
C_{i,0}^{+*}=V_i\\
C_{i,t}^{+*}=\max(0,C_{i,t-1}^{+*}+X_{i,t}-\mu_{i,0}-k_i)), \text{ for } t \geq 1
\end{cases}
\end{equation}
where $V_i$ is randomly drawn with replacement from $\{C_{j,2000}^+\}_{j=1}^{10^6}$. Since the above modified CUSUM statistic starts from the steady-state, the null distribution of the $C_{i,t}^{+*}$ at any time point $t$ follows the null steady-state distribution. As a result, we can utilize the close-form formula provided in Grigg and Spiegelhalter (2008) to calculate the $p$-values of the $C_{i,t}^{+*}$ quickly. For simplicity, we abuse the notation and also denote those $p$-values by $p_{i,t}$. Based on the above modified CUSUM statistics and their corresponding $p$-values $p_{i,t}$, we calculate $G_t$ in (\ref{eqn:Gt}). Our monitoring scheme in the first stage is then to monitor those $G_t$ from our modified CUSUM statistics and signals an alarm if $G_t>h$, where $h$ is the control limit and can be determined through Monte-Carlo simulation to satisfy the IC-ARL requirement.

If we signal an alarm, i.e., $G_t>h$, in our first stage, we move to our second stage to identify which data stream is OC. In this stage, we first need to identify the appropriate test statistic $W_{i,t}$ for the hypothesis testing problem (\ref{test0}) for each data stream. Due to the optimality of the CUSUM statistic, we can use the above modified $C^{+*}_{i,t}$  as our $W_{i,t}$. Alternatively we can also use $1-p_{i,t}$, where $p_{i,t}$ is again the $p$-value of $C^{+*}_{i,t}$, as our $W_{i,t}$. Since $1-p_{i,t}$ is simply a non-decreasing function of $C^{+*}_{i,t}$, these two test statistics are equivalent. Given the fact that the $p_{i,t}$ have been calculated in our first stage, we use  $1-p_{i,t}$ as our $W_{i,t}$ in this example. Based on $1-p_{i,t}$ as our $W_{i,t}$, we next use Monte-Carlo simulation and the bisection method as described in Section 3 to determine the control limit $c_h$ of our second stage to satisfy the PCER requirement, that is, finding $c_h$ such that
\[
\frac{1}{m}\sum_{i=1}^{m}P_{H_{0,t}}\Big(1-p_{i,t} > c_h|G_t>h\Big)=\alpha,
\]
where $\alpha$ is the desired PCER. Recall that, in order to easily use Monte Carlo simulation and the bisection method to determine the above control limit $c_h$, we need to assume that both $G_t$ and $W_{i,t}$ ($i=1,...,m$) follow their respective null steady-state distributions when $H_{0,t}$ is true. Therefore, our modification of the CUSUM statistics described earlier is also necessary for this assumption to hold. Once we obtain the control limit $c_h$ using Monte Carlo simulation and the bisection method, we identify the $i$-th data stream as the OC data stream if $1-p_{i,t} > c_h$, and as the IC data stream otherwise.

\section{Simulation Study}
In this section, we refer to our two-stage procedure described in Section 4 as the two-stage CUSUM procedure. In the following, we present a simulation study to compare the performance of this two-stage CUSUM procedure with the one using the methodology proposed in Li and Tsung (2009, 2012).

Using the methodology proposed in Li and Tsung (2009, 2012) in the setting considered in Section 4, at each time point $t$, we will monitor each data stream using the CUSUM statistic and calculate its $p$-value. To use the close-form formula provided in Grigg and Spiegelhalter (2008) to calculate those $p$-values quickly, we also use the modified CUSUM statistic $C_{i,t}^{+*}$ defined in (\ref{eqn:CUSUM1}) instead of the original CUSUM statistic $C_{i,t}^{+}$ in (\ref{eqn:CUSUM}) to monitor each data stream. The $p$-values of the  $C_{i,t}^{+*}$ are again denoted by $p_{i,t}$, and their corresponding order statistics are $p_{(1),t}\leq \cdots \leq p_{(m),t}$. Let
\[
I=\max\{k: p_{(k),t} \leq kq/m,\quad 1\leq k \leq m\},
\]
where $q$ is the point-wise FDR and is determined by Monte Carlo simulation to satisfy the desired IC-ARL. If $I=0$, all the data streams are IC. If $I>0$, all the data streams associated with $p_{(1),t},...,p_{(I),t}$ are classified as the OC data streams and the rest are the IC data streams. Hereafter, we refer to this procedure as the LT procedure.

Since both our two-stage CUSUM procedure and the LT procedure can be designed to have the desired IC-ARL control, we compare these two procedures based on how quickly they can correctly identify all the OC data streams. If there is only one  data stream, this can be easily measured by the OC-ARL, which is the average number of the observations collected after the system becomes OC at the time when the scheme triggers an alarm. However, when there are a large number of data streams, even if there is only one OC data stream, it is not guaranteed that the scheme can identify the right OC data stream at the first time when it triggers an alarm. When there are multiple OC data streams, it is also quite often that the scheme will not be able to identify all the OC data streams at once. Therefore, using the OC-ARL might not be able to reflect well the performance when monitoring high-dimensional data streams.

To find an appropriate criterion to evaluate the performance, different from the simulations conducted to obtain the OC-ARL where the monitoring is stopped at the first time when the scheme triggers an alarm, in our simulations we assume that after the OC data stream is identified by the scheme, the monitoring of that data stream continues, the observations collected afterwards from that data stream will be drawn from its IC distribution, and its CUSUM statistic will be restarted at the value randomly drawn with replacement from $\{C_{j,2000}^+\}_{j=1}^{10^6}$ as described in the previous section. When all the OC data streams are identified, the monitoring is stopped and we calculate the average time to detect all the OC data streams, which we denote by the ATDOC. Comparing with the OC-ARL, the ATDOC can better measure how effective the monitoring scheme is in terms of detecting all the OC data streams. Therefore, in the following we use the ATDOC criterion to compare the performance of the two monitoring schemes.

The simulation settings we consider here are similar to those used in Zou et al. (2014). The general settings are the following. Among the $m$ data streams, $m_0$ data streams are from the IC distribution $N(0,1)$, the remaining $m_1=m-m_0$ data streams are from some OC distribution $N(\delta_i,1)$, $i=1,...,m_1$. Two scenarios of the OC distribution are considered: (I) equal allocation, i.e., $\delta_i=\mu$; (II) increasing allocation, i.e., $\delta_i=\mu\log(1+\sqrt{i})$, $i=1,...,m_1$. In both scenarios,  $\mu$ is the target OC mean value specified in the CUSUM statistic, and throughout the simulations, we set $\mu=0.5$.

In the first simulation study, we choose $m=100$. The desired IC-ARL is set at 200, 500, 1000 or 10000. Besides the IC-ARL, we also need to specify the desired PCER for our two-stage CUSUM procedure. Three scenarios are considered: PCER=0.01, 0.03 or 0.05. The control limits $h$ in the first stage and $c_h$ in the second stage for our two-stage CUSUM procedure along with the point-wise FDR $q$ for the LT procedure can be obtained through Monte-Carlo simulation. The results are presented in Table \ref{tab:CL100}.
\begin{table}[!hbtp]
\begin{center}
\caption{The control limits used in our proposed two-stage CUSUM and LT procedures when $m=100$, $\mu=0.5$.}\label{tab:CL100}
\begin{tabular}{|c|c|c|c|c||c|c|c|c|c|}
  \hline
IC-ARL & $q$ & $h$ & PCER & $c_h$ & IC-ARL & $q$ & $h$ & PCER & $c_h$ \\
  \hline
200 & .04865 & 16.546 & .01 & .99577 & 500 & .02087  & 22.900 & .01 &  .99716\\
& & & .03  & .98428 & & & & .03 & .98710\\
& & & .05 & .97039 &  & & & .05 & .97472\\
\hline
1000 & .01034  & 28.570 & .01 & .99802 & 10000 & .00108  & 49.119 & .01 & .99947\\
& & & .03  & .98863 & & & & .03 & .99225\\
& & & .05 & .97746 &  & & & .05 & .98097\\
  \hline
\end{tabular}
\end{center}
\end{table}

Using those control limits, we apply our two-stage CUSUM procedure and the LT procedure to different OC settings. Tables \ref{tab:ATDOC1} and \ref{tab:ATDOC2} reports the the average and standard deviation of the ATDOC of the two procedures from 1000 simulations for those settings. As we can see, in all the settings considered here, for our two-stage CUSUM procedure, higher PCERs lead to smaller ATDOCs. So if we want to detect the OC data streams faster, we can increase the PCER that we can tolerate. In contrast, the LT procedure does not have this kind of flexibility. Comparing between these two procedures, our two-stage CUSUM procedure has smaller ATDOC than the LT procedure for all the settings except for the case of $m_1=1$, and the advantage of our two-stage CUSUM procedure over the LT procedure becomes more pronounced when the number of the OC data streams increases. When PCER=0.05, our two-stage CUSUM procedure can cut the ATDOC of the LT procedure by half in many cases. In the case of $m_1=1$, since there is only one OC data stream, the ATDOC is roughly equal to the OC-ARL. Based on our simulation, the LT procedure has very similar OC-ARL performance as the one using the global monitoring statistic $\max_{i=1,...,m}C^+_{i,t}$. As mentioned earlier, the monitoring scheme using $\max_{i=1,...,m}C^+_{i,t}$ is optimal when there is exactly one OC data stream. This explains why the LT procedure performs better than the two-stage CUSUM procedure when $m_1=1$.

\begin{table}[!htpb]
\begin{center}
\caption{The ATDOC comparison between our proposed two-stage CUSUM procedure and the LT procedure when $m=100$. Standard deviations of the ATDOC are in parentheses.}\label{tab:ATDOC1}
\begin{tabular}{|c|c|c|c|ccc|}
  \hline
& & &  & \multicolumn{3}{|c|}{ Two-stage CUSUM} \\
& IC-ARL & $m_1$ & LT & PCER=.01 & PCER=.03 & PCER=.05 \\
\hline
& & 1 & 50.4 (25.7) & 57.2 (27.6) & 55.4 (27.3) & 54.0 (27.6) \\
& & 3 & 49.7 (15.1) & 49.5 (16.8) & 44.4 (16.0) & 41.9 (15.5) \\
& & 5 & 49.4 (11.8) & 47.0 (12.4) & 41.4 (11.9) & 38.8 (11.8) \\
& & 8 & 48.0 (9.2) & 43.5 (9.2) & 37.5 (8.6) & 34.7 (8.4) \\
& 200 & 10 & 48.0 (8.3) & 42.2 (8.2) & 36.3 (7.6) & 33.2 (7.3) \\
& & 20 & 45.2 (5.9) & 38.2 (5.3) & 31.5 (5.00) & 28.6 (4.8) \\
& & 50 & 41.1 (3.7) & 35.9 (3.1) & 27.1 (2.7) & 23.5 (2.6) \\
& & 80 & 38.1 (2.9) & 35.2 (2.4) & 25.8 (2.0) & 21.7 (1.9) \\
Equal & & 100  & 36.6 (2.6) & 35.1 (2.1) & 25.5 (1.8) & 21.2 (1.7) \\
\cline{2-7}
allocation & &  1 & 58.4 (27.7) & 67.8 (30.7) & 66.5 (31.2) & 65.5 (31.5) \\
& & 3 & 58.0 (17.4) & 56.0 (18.6) & 50.5 (18.2) & 47.2 (17.9) \\
& & 5 & 56.7 (12.7) & 51.5 (13.1) & 45.6 (12.6) & 42.4 (12.1) \\
& & 8 & 55.6 (10.7) & 47.8 (10.6) & 41.0 (10.0) & 37.8 (9.4) \\
& 500 & 10 & 55.1 (9.1) & 46.1 (8.7) & 39.2 (8.0) & 35.9 (7.7) \\
& &  20 & 52.5 (6.7) & 41.6 (5.8) & 34.0 (5.4) & 31.0 (5.3) \\
& & 50 & 48.6 (4.2) & 39.3 (3.3) & 29.1 (2.9) & 25.2 (2.8) \\
& & 80 & 45.7 (3.4) & 38.5 (2.6) & 27.5 (2.1) & 23.1 (2.0) \\
& & 100 & 44.3 (2.8) & 38.3 (2.3) & 27.2 (1.9) & 22.6 (1.7) \\
\hline
 \hline
& & 1 & 101.4 (69.1) & 116.1 (76.7) & 106.9 (70.2) & 103.5 (68.8) \\
& &  3 & 72.1 (28.1) & 73.2 (30.0) & 64.8 (28.7) & 59.6 (27.8) \\
& &  5 & 59.4 (19.1) & 57.8 (20.1) & 50.4 (18.9) & 46.3 (17.8) \\
& &  8 & 48.5 (11.7) & 44.6 (11.7) & 38.6 (11.1) & 35.7 (11.0) \\
& 200 & 10 & 44.2 (9.9) & 40.1 (10.3) & 34.2 (9.5) & 31.0 (8.8) \\
& & 20 & 31.6 (5.0) & 28.1 (5.2) & 23.1 (4.6) & 21.1 (4.6) \\
& & 50 & 20.6 (2.4) & 19.2 (2.2) & 14.8 (1.9) & 12.8 (1.9) \\
& & 80 & 16.1 (1.4) & 16.2 (1.4) & 12.3 (1.3) & 10.4 (1.2) \\
Increasing & & 100 & 14.4 (1.2) & 15.1 (1.1) & 11.4 (1.0) & 9.6 (1.0) \\
 \cline{2-7}
allocation & & 1 & 121.1 (81.6) & 138.0 (91.1) & 131.8 (86.7) & 129.0 (85.2) \\
& & 3 & 83.3 (29.8) & 83.0 (33.1) & 73.1 (30.6) & 68.3 (29.5) \\
& & 5 & 68.7 (20.7) & 63.9 (22.2) & 56.0 (21.2) & 51.2 (19.5) \\
& & 8 & 55.8 (13.3) & 49.6 (14.4) & 42.0 (13.3) & 37.9 (11.8) \\
& 500 & 10 & 50.8 (10.8) & 44.1 (11.3) & 37.4 (10.5) & 34.0 (10.1) \\
& &  20 & 36.8 (5.4) & 30.9 (5.9) & 25.2 (5.5) & 22.8 (5.2) \\
& & 50 & 24.2 (2.5) & 21.0 (2.4) & 15.9 (2.2) & 13.8 (2.2) \\
& & 80 & 19.3 (1.6) & 17.6 (1.5) & 13.1 (1.4) & 11.1 (1.3) \\
& & 100 & 17.3 (1.3) & 16.4 (1.3) & 12.1 (1.2) & 10.2 (1.1) \\
\hline
\end{tabular}
\end{center}
\end{table}

\begin{table}[!htpb]
\begin{center}
\caption{The ATDOC comparison between our proposed two-stage CUSUM procedure and the LT procedure when $m=100$. Standard deviations of the ATDOC are in parentheses.}\label{tab:ATDOC2}
\begin{tabular}{|c|c|c|c|ccc|}
  \hline
& & &  & \multicolumn{3}{|c|}{ Two-stage CUSUM} \\
& IC-ARL & $m_1$ & LT & PCER=.01 & PCER=.03 & PCER=.05 \\
\hline
& & 1 & 64.1 (30.0) & 73.5 (32.2) & 72.8 (32.2) & 72.2 (32.0) \\
& &  3 & 63.4 (17.2) & 60.3 (18.6) & 52.7 (17.9) & 50.1 (17.4) \\
& &  5 & 63.0 (13.4) & 55.7 (14.1) & 48.1 (13.2) & 45.4 (13.2) \\
& &  8 & 61.9 (11.3) & 51.7 (11.2) & 43.8 (10.5) & 40.5 (10.1) \\
& 1000&  10 & 61.0 (9.5) & 49.6 (9.2) & 41.7 (8.6) & 38.3 (8.1) \\
& &  20 & 59.1 (6.7) & 45.2 (5.9) & 36.4 (5.5) & 33.1 (5.2) \\
& &  50 & 54.5 (4.2) & 42.0 (3.3) & 30.2 (2.8) & 26.4 (2.8) \\
& &  80 & 52.0 (3.5) & 41.6 (2.8) & 28.8 (2.3) & 24.3 (2.2) \\
Equal & &  100 & 50.4 (3.0) & 41.2 (2.4) & 28.3 (1.9) & 23.6 (1.8) \\
  \cline{2-7}
  allocation & &  1 & 83.0 (36.9) & 89.9 (39.0) & 89.5 (39.2) & 89.4 (39.1) \\
& & 3 & 80.8 (19.6) & 72.1 (20.0) & 61.7 (20.5) & 58.1 (19.4) \\
& & 5 & 80.7 (15.8) & 66.9 (15.8) & 55.5 (14.5) & 51.4 (14.4) \\
& & 8 & 79.9 (12.1) & 62.1 (11.2) & 50.0 (10.5) & 46.0 (10.1) \\
& 10000 & 10 & 79.0 (10.7) & 59.2 (9.8) & 47.1 (9.1) & 42.7 (8.8) \\
& & 20 & 77.6 (8.1) & 55.0 (6.8) & 41.1 (6.3) & 36.8 (5.9) \\
& & 50 & 73.3 (4.9) & 52.4 (4.0) & 34.1 (3.3) & 29.3 (3.1) \\
& & 80 & 71.2 (4.0) & 52.0 (3.0) & 32.4 (2.4) & 26.5 (2.2) \\
& & 100 & 69.9 (3.4) & 51.8 (2.7) & 31.8 (2.1) & 25.7 (1.9) \\
\hline
 \hline
& & 1 & 138.9 (95.8) & 155.7 (102.7) & 152.2 (100.3) & 150.2 (100.0) \\
& & 3 & 93.4 (34.9) & 91.3 (38.9) & 81.3 (38.3) & 75.8 (37.1) \\
& & 5 & 76.5 (22.1) & 70.4 (23.7) & 59.9 (22.1) & 55.6 (21.2) \\
& & 8 & 61.7 (14.1) & 53.6 (14.8) & 45.0 (14.0) & 41.4 (13.6) \\
& 1000 & 10 & 55.8 (11.4) & 47.3 (12.3) & 39.5 (11.3) & 36.1 (10.8) \\
& & 20 & 41.3 (6.3) & 33.5 (6.5) & 26.9 (5.9) & 24.2 (5.5) \\
& & 50 & 27.2 (2.6) & 22.5 (2.5) & 16.7 (2.4) & 14.4 (2.3) \\
& & 80 & 22.0 (1.8) & 19.0 (1.7) & 13.7 (1.6) & 11.7 (1.5) \\
Increasing & & 100 & 19.8 (1.4) & 17.6 (1.3) & 12.6 (1.2) & 10.7 (1.2) \\
 \cline{2-7}
allocation & &1 & 187.3 (120.6) & 199.3 (123.1) & 196.6 (121.2) & 196.4 (121.3) \\
& &  3 & 121.6 (40.8) & 112.2 (45.0) & 94.8 (43.9) & 87.3 (40.6) \\
& &  5 & 98.8 (26.4) & 85.4 (27.7) & 71.2 (26.2) & 65.5 (25.1) \\
& &  8 & 81.6 (17.1) & 67.1 (17.7) & 53.6 (16.1) & 48.4 (15.8) \\
& 10000 &  10 & 74.0 (14.4) & 59.0 (15.1) & 47.2 (15.0) & 42.2 (14.3) \\
& &  20 & 54.3 (7.6) & 41.2 (7.5) & 30.8 (7.2) & 27.5 (6.8) \\
& &  50 & 36.8 (3.0) & 28.2 (3.0) & 19.1 (2.8) & 16.2 (2.7) \\
& &  80 & 30.1 (1.9) & 23.5 (1.9) & 15.5 (1.8) & 12.8 (1.6) \\
& &  100 & 27.4 (1.6) & 21.8 (1.5) & 14.1 (1.4) & 11.6 (1.3) \\

\hline
\end{tabular}
\end{center}
\end{table}

We also report in Tables \ref{tab:global1} and \ref{tab:global2} the simulated global FDR of the LT procedure and the simulated global PCER for our two-stage CUSUM procedure in the settings considered above. As we can see from the table, the global FDR of the LT procedure can be much higher than the point-wise FDR $q$. This shows that control of the point-wise FDR does not guarantee control of the global FDR. In contrast, our two-stage CUSUM procedure remains control of the global PCER in all the settings.
\begin{table}[!htpb]
\begin{center}
\caption{Global FDR  of the LT procedure and global PCER of our two-stage CUSUM procedure when $m=100$.}\label{tab:global1}
\begin{tabular}{|c|c|c|c|ccc|}
  \hline
& & & LT  & \multicolumn{3}{|c|}{ Two-stage CUSUM} \\
& & & Point-wise FDR & \multicolumn{3}{|c|}{Nominal PCER}\\
IC-ARL & & $m_1$ &  $q=.04865$ & .01 & .03 & .05 \\
\hline
& & 1 & .15868 & .00553 & .01889 & .03427 \\
& & 3 & .11974 & .00440 & .01645 & .03128 \\
& & 5 & .09631 & .00364 & .01495 & .02911 \\
& Equal & 8 & .08320 & .00275 & .01340 & .02607 \\
& allocation & 10 & .06626 & .00234 & .01258 & .02585 \\
& & 20 & .04953 & .00113 & .00865 & .02014 \\
& & 50 & .03012 & .00062 & .00331 & .00882 \\
& & 80 & .02231 & .00052 & .00251 & .00589 \\
200 & & 100 & .01772 & .00046 & .00221 & .00500 \\
  \cline{2-7}
& & 1 & .24242 & .00609 & .02053 & .03754 \\
& & 3 & .15774 & .00464 & .01703 & .03190 \\
& & 5 & .11862 & .00395 & .01585 & .03040 \\
& Increasing &  8 & .08587 & .00296 & .01352 & .02755 \\
& allocation & 10 & .07876 & .00248 & .01256 & .02545 \\
& & 20 & .04791 & .00122 & .00784 & .01882 \\
& & 50 & .02875 & .00090 & .00404 & .00915 \\
& & 80 & .02166 & .00083 & .00361 & .00769 \\
& & 100 & .01608 & .00076 & .00334 & .00709 \\
  \hline
  \hline
&  & & LT  & \multicolumn{3}{|c|}{ Two-stage CUSUM} \\
& & & Point-wise FDR & \multicolumn{3}{|c|}{Nominal PCER}\\
IC-ARL & & $m_1$ &  $q=.02087 $ & .01 & .03 & .05 \\
\hline
& & 1 & .08733 & .00391 & .01534 & .02887 \\
& & 3 & .05750 & .00308 & .01404 & .02706 \\
& & 5 & .04357 & .00260 & .01273 & .02539 \\
& Equal & 8 & .03581 & .00205 & .01173 & .02351 \\
& allocation & 10 & .02983 & .00161 & .01063 & .02271 \\
& & 20 & .02168 & .00071 & .00731 & .01800 \\
& & 50 & .01356 & .00042 & .00273 & .00769 \\
& & 80 & .00963 & .00034 & .00200 & .00488 \\
500 & & 100 & .00746 & .00030 & .00174 & .00416 \\
\cline{2-7}
& & 1 & .12683 & .00430 & .01784 & .03204 \\
& & 3 & .08237 & .00343 & .01449 & .02805 \\
& & 5 & .06009 & .00303 & .01409 & .02656 \\
& Increasing & 8 & .04000 & .00219 & .01181 & .02403 \\
& allocation & 10  & .03803 & .00187 & .01135 & .02300 \\
& & 20 & .02527 & .00086 & .00727 & .01742 \\
& & 50 & .01252 & .00058 & .00326 & .00773 \\
& & 80 & .00915 & .00055 & .00287 & .00642 \\
& & 100 & .00758 & .00053 & .00270 & .00585 \\
\hline\end{tabular}
\end{center}
\end{table}

\begin{table}[!htpb]
\begin{center}
\caption{Global FDR  of the LT procedure and global PCER of our two-stage CUSUM procedure when $m=100$.}\label{tab:global2}
\begin{tabular}{|c|c|c|c|ccc|}
  \hline
& & & LT  & \multicolumn{3}{|c|}{ Two-stage CUSUM} \\
& & & Point-wise FDR & \multicolumn{3}{|c|}{Nominal PCER}\\
IC-ARL & & $m_1$ &  $q=.01034$ & .01 & .03 & .05 \\
\hline
& &1 & .04550 & .00270 & .01335 & .02480 \\
& & 3 & .02590 & .00219 & .01226 & .02383 \\
& & 5 & .02261 & .00194 & .01160 & .02290 \\
& Equal & 8 & .01920 & .00150 & .01060 & .02148 \\
& allocation & 10 & .01481 & .00114 & .00990 & .02064 \\
& & 20 & .01120 & .00050 & .00697 & .01629 \\
& & 50 & .00723 & .00028 & .00245 & .00708 \\
& & 80 & .00492 & .00023 & .00178 & .00431 \\
1000 & & 100 & .00381 & .00021 & .00153 & .00374 \\
  \cline{2-7}
& &1 & .07483 & .00359 & .01550 & .02860 \\
& & 3 & .04510 & .00275 & .01339 & .02569 \\
& & 5 & .02915 & .00213 & .01235 & .02397 \\
& Increasing & 8 & .02568 & .00169 & .01149 & .02226 \\
& allocation & 10 & .02114 & .00137 & .01034 & .02074 \\
& & 20 & .01262 & .00057 & .00668 & .01574 \\
& & 50 & .00634 & .00041 & .00294 & .00687 \\
& & 80 & .00450 & .00038 & .00246 & .00566 \\
& & 100 & .00373 & .00034 & .00236 & .00525 \\
  \hline
  \hline
&  & & LT  & \multicolumn{3}{|c|}{ Two-stage CUSUM} \\
& & & Point-wise FDR & \multicolumn{3}{|c|}{Nominal PCER}\\
IC-ARL & & $m_1$ &  $q=.00108 $ & .01 & .03 & .05 \\
\hline
& &1 & .00750 & .00101 & .00909 & .02071 \\
& & 3 & .00550 & .00064 & .00808 & .01943 \\
& & 5 & .00200 & .00055 & .00786 & .01946 \\
& Equal & 8 & .00289 & .00039 & .00748 & .01848 \\
& allocation &10 & .00218 & .00031 & .00706 & .01765 \\
& & 20 & .00114 & .00012 & .00531 & .01557 \\
& & 50 & .00057 & .00007 & .00167 & .00719 \\
& & 80 & .00059 & .00005 & .00115 & .00379 \\
10000 & & 100 & .00036 & .00005 & .00101 & .00315 \\
\cline{2-7}
& &1 & .01000 & .00136 & .01064 & .02247 \\
& &  3 & .00675 & .00084 & .00895 & .02094 \\
& &  5 & .00317 & .00067 & .00851 & .02009 \\
&Increasing &  8 & .00178 & .00048 & .00762 & .01831 \\
& allocation &  10 & .00227 & .00035 & .00749 & .01815 \\
& &  20 & .00100 & .00014 & .00520 & .01488 \\
& &  50 & .00073 & .00011 & .00203 & .00639 \\
& &  80 & .00057 & .00010 & .00170 & .00489 \\
& &  100 & .00050 & .00008 & .00157 & .00438 \\

\hline\end{tabular}
\end{center}
\end{table}

We also conduct a simulation study with a higher dimension $m=1000$. For our two-stage procedure, PCER is set at 0.005, 0.01 or 0.02. The control limits $h$ and $c_h$ for our two-stage CUSUM procedure along with the point-wise FDR $q$ for the LT procedure obtained through Monte-Carlo simulation for IC-ARL=1000 are listed in Table \ref{tab:CL1000}.
\begin{table}[!htpb]
\begin{center}
\caption{The control limits used in our proposed two-stage CUSUM and LT procedures when $m=1000$.}\label{tab:CL1000}
\begin{tabular}{|c|c|c|c|c|}
  \hline
IC-ARL & $q$ & $h$ & PCER & $c_h$ \\
  \hline
1000 & .01045  & 32.593 & .005 & .99737\\
& & &  .01 & .99379\\
& & & .02 & .98576\\
  \hline
\end{tabular}
\end{center}
\end{table}
The ATDOC comparison between the two-stage CUSUM procedure and the LT procedure is reported  in Figure \ref{fig:ATDOC}. Similar to the $m=100$ case, our two-stage CUSUM procedure has much smaller ATDOC than the LT procedure in all the cases except the case of $m_1=1$. The explanation is the same as the one given above.
Figure \ref{fig:FDR} shows the global FDR control of the LT procedure and the global PCER control of our two-stage CUSUM procedure. In both figures,  the ratio of the global FDR and the point-wise FDR is plotted for the LT procedure, and the ratio of the global PCER and its nominal value is plotted for our two-stage CUSUM procedure. As we can see, the global FDR can be much higher than the point-wise FDR in the LT procedure, while our two-stage CUSUM procedure has good control of the global PCER.

\begin{figure}[!htpb]
\begin{center}
\begin{tabular}{c}
\includegraphics[width=3.0in,height=3.3in]{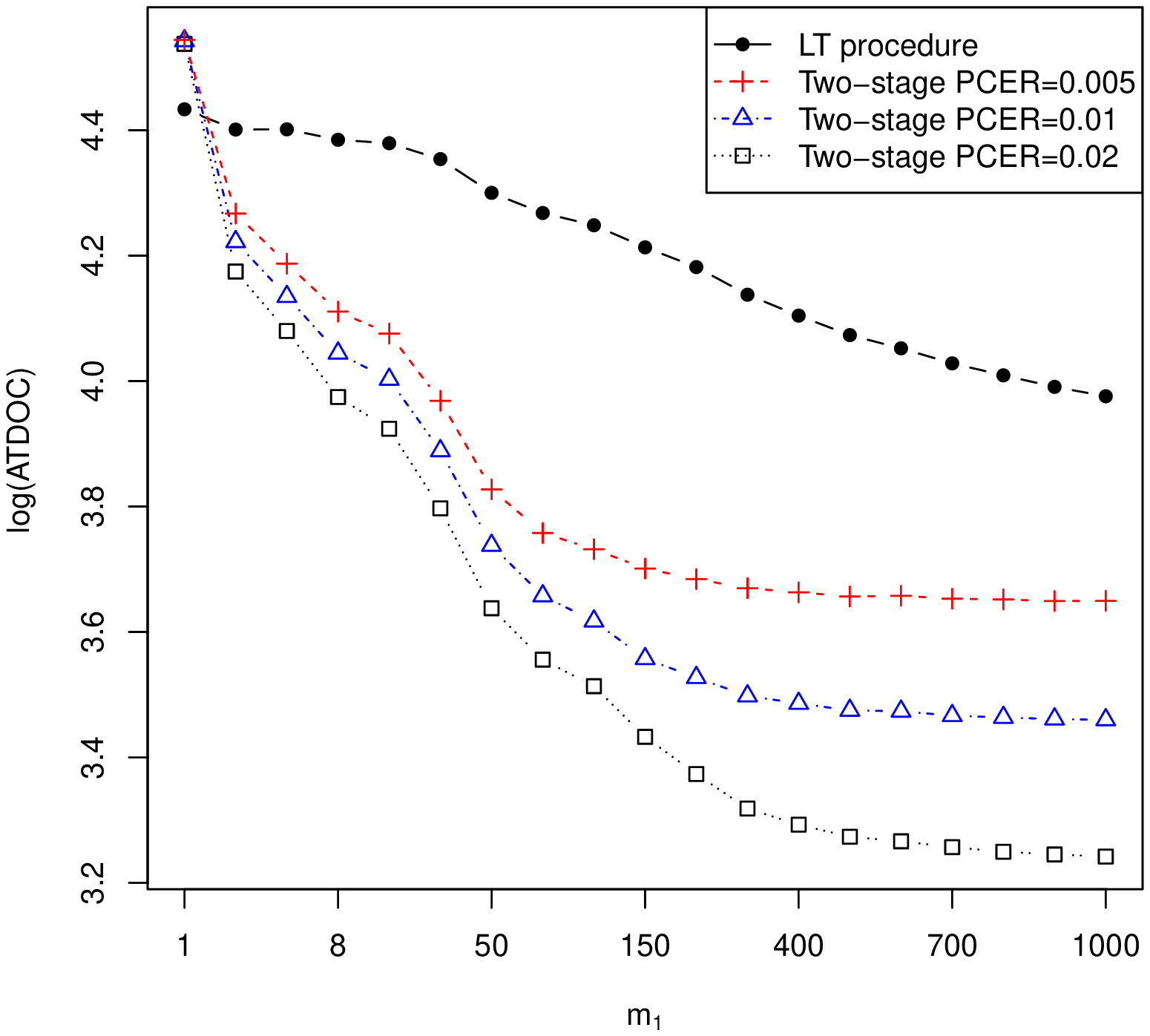}\\(a)
\end{tabular}
\begin{tabular}{c}
\includegraphics[width=3.0in,height=3.3in]{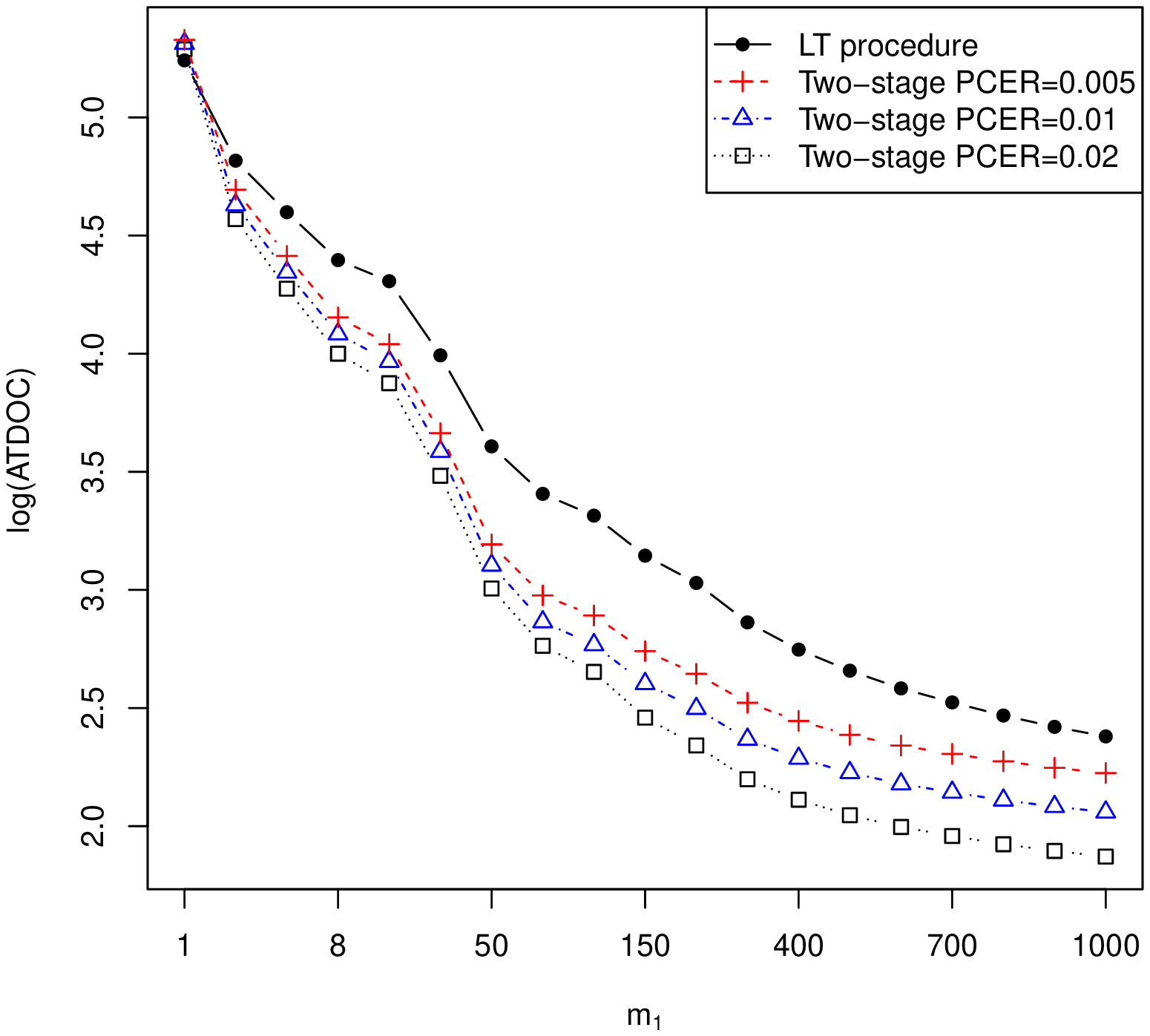}\\(b)
\end{tabular}
\caption{The ATDOC comparison between our proposed two-stage CUSUM procedure and the LT procedure when $m=1000$ and IC-ARL=1000 for: (a) equal allocation; (b) increasing allocation. Both figures are plotted on the log scale.}\label{fig:ATDOC}
\end{center}
\end{figure}

\begin{figure}[!htpb]
\begin{center}
\begin{tabular}{c}
\includegraphics[width=3.0in,height=3.3in]{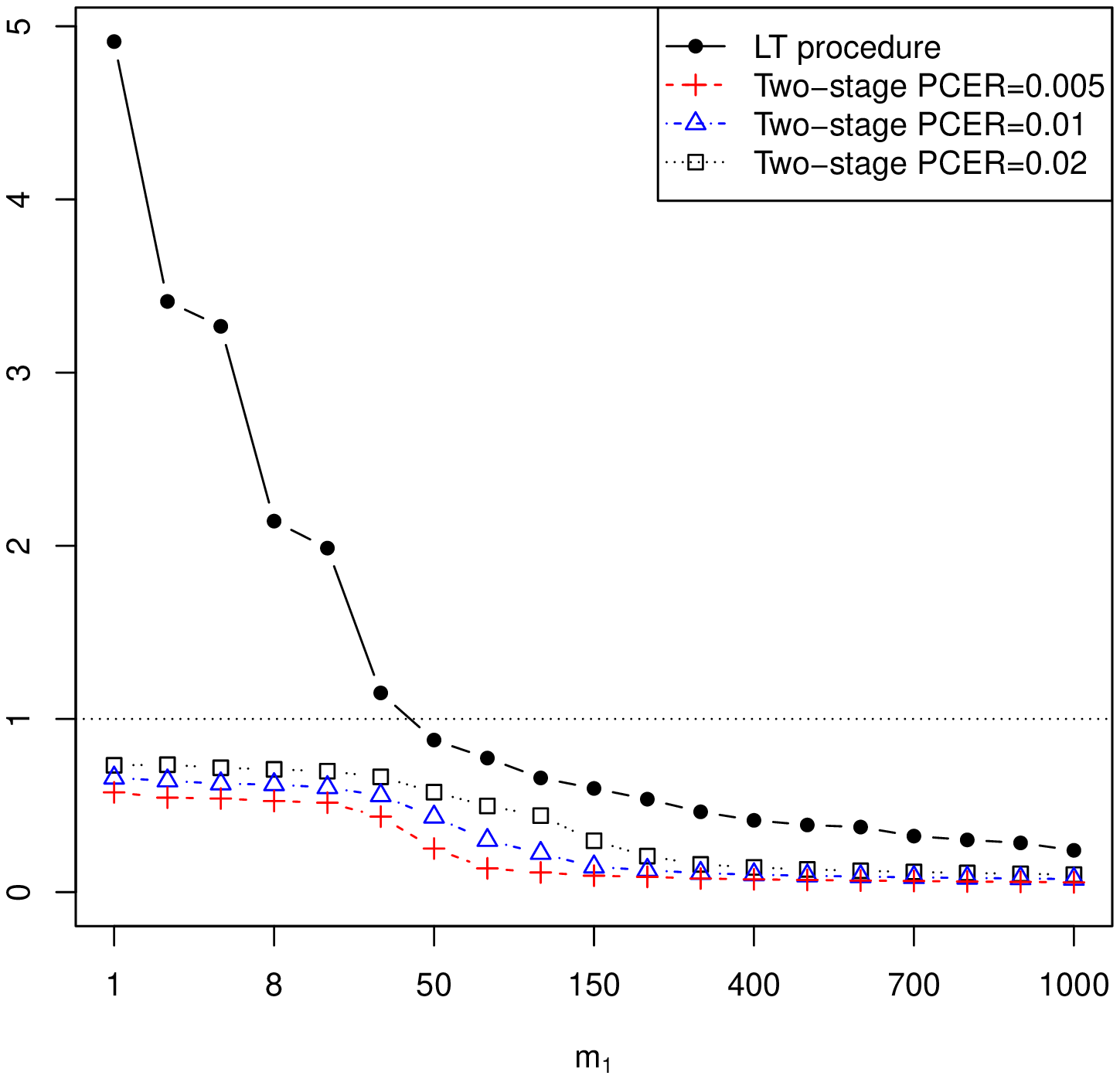}\\(a)
\end{tabular}
\begin{tabular}{c}
\includegraphics[width=3.0in,height=3.3in]{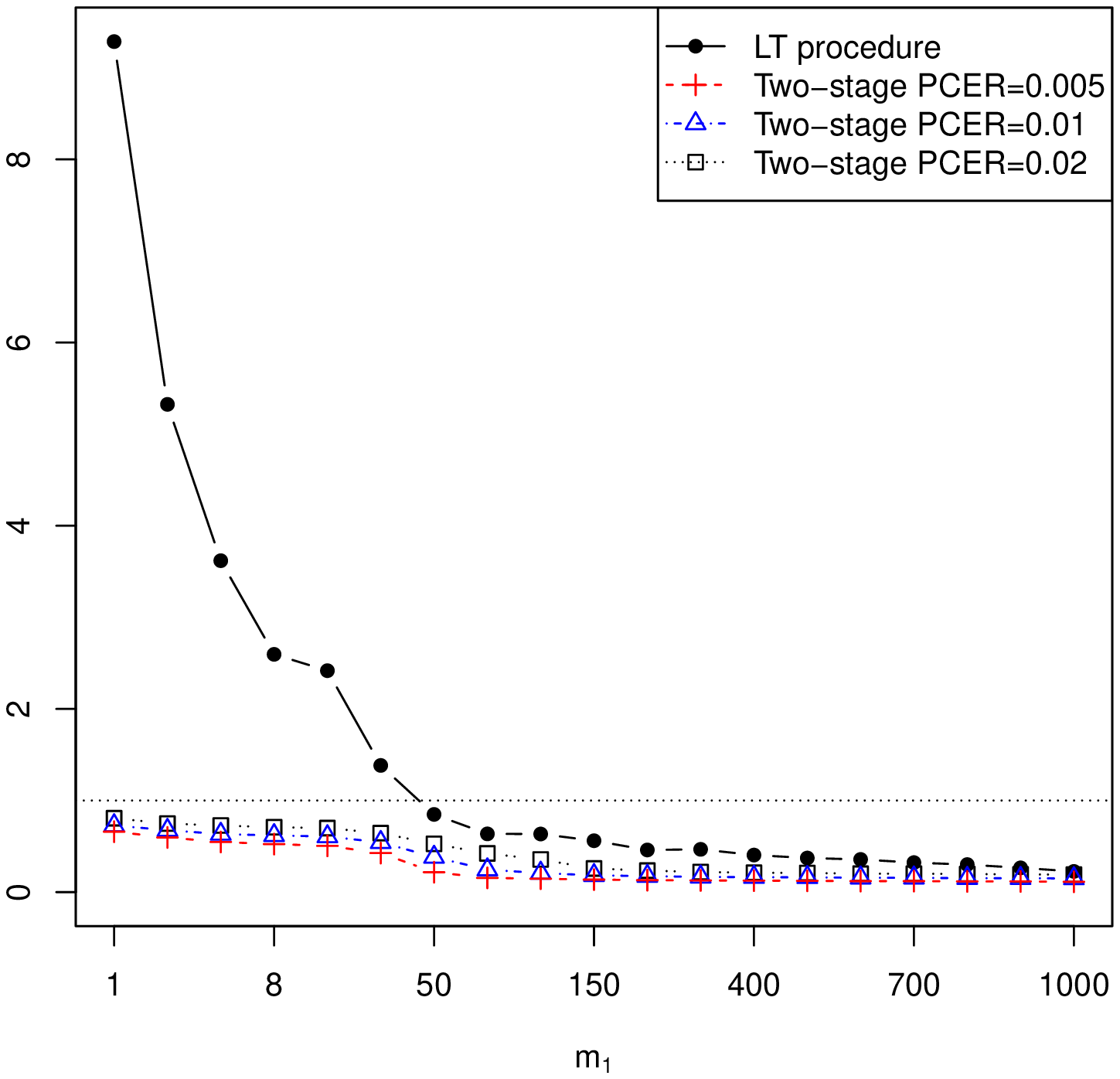}\\(b)
\end{tabular}
\caption{The ratio of the global FDR and the point-wise FDR for the LT procedure and the ratio of the global PCER and its nominal value for our two-stage CUSUM procedure when $m=1000$ and IC-ARL=1000 for: (a) equal allocation; (b) increasing allocation.}\label{fig:FDR}
\end{center}
\end{figure}

\section{Real data application}
In this section, we use a real data set to demonstrate the application of our proposed two-stage procedure. The data set we use was collected from a modern semiconductor manufacturing process and is publicly available in the UC Irvine Machine Learning Repository (http://archive.ics.uci.edu/ml/datasets/SECOM). It contains a total of $1567$ vector observations, and each vector observation consists of $591$ measured features. There is a huge amount of constant values across the $1567$ observations in $191$ features, therefore we exclude them from our analysis, which leaves us a data set with $1567$ vector observations and each being a $400$-dimensional vector.  To demonstrate the application of our proposed two-stage monitoring procedure, we treat the $1567$ observations of each feature as one data stream, and our task is then to monitor the $m=400$ data streams simultaneously and identify any non-conforming observations as soon as possible.

Among the $1567$ vector observations, $104$ of them are labeled as nonconforming and the remaining $1463$ conforming. For illustration, we use all the conforming observations (the $1463$ vector observations) as the historical sample and apply our proposed two-stage monitoring procedure to the nonconforming ones (the $104$ vector observations) to see how many nonconforming observations our procedure is able to detect. For simplicity of notation, we reorganize the data set such that the conforming observations from the $i$-th feature/data stream are denoted by $X_{i,t}$, $t=1,...,1463$ and the nonconforming ones $X_{i,t}$, $t=1464,...,1567$.

Since not all the data are normally distributed, to use the global monitoring statistic $G_t$ in (\ref{eqn:Gt}), we apply the inverse transformation, $\Phi^{-1}(\hat{F}_{i}(X_{i,t})$, $i=1,...,400$, $t=1,...,1567$ to each data stream, where $\hat{F}_{i}$ is the empirical distribution function based on the $1463$ historical observations of the $i$-th data stream. Under this transformation, the in-control distributions of all the data streams are approximately $N(0,1)$. Denote the transformed nonconforming observations by  $X^*_{i,t}$, $i=1,...,400$, $t=1464,...,1567$. For the $i$-th data stream and at each $t$, $t=1464,...,1567$, we then calculate our CUSUM statistic $C^{+*}_{i,t}$ in (\ref{eqn:CUSUM1}) based on those nonconforming observations $X^*_{i,t}$ with $\mu_{i,0}=0$ and $k_i=0.25$.  By applying the close-form $p$-value formula from Grigg and Spiegelhalter (2008) to those CUSUM statistics $C^{+*}_{i,t}$, we obtain their marginal $p$-values, $p_{i,t}$. According to Zou et al. (2014), although there exists certain weak positive dependence structure in this data set, their proposed global monitoring statistic $G_t$ in (\ref{eqn:Gt}) can still be effective. Therefore, in our proposed two-stage procedure,  we still use this $G_t$ as our global monitoring statistic and calculate its value based on $\{p_{i,t}\}_{i=1}^{400}$ to decide whether an alarm should be given at $t$. If $G_t$ is greater than the control limit $h$, we then compare $1-p_{i,t}$ to the appropriate control limit $c_h$ to decide whether the data stream is nonconforming. In the above two-stage procedure, we set the nominal IC-ARL to be 1000, and the corresponding control limits $h$ and $c_h$ used in the two stages are listed in Table \ref{tab:CL400}. As a comparison, we also apply the LT procedure to this data set and its control limit $q$ for the same nominal IC-ARL is also reported in Table \ref{tab:CL400}.
\begin{table}[!htpb]
\begin{center}
\caption{The control limits used in our proposed two-stage CUSUM and LT procedures when $m=400$.}\label{tab:CL400}
\begin{tabular}{|c|c|c|c|c|}
  \hline
IC-ARL & $q$ & $h$ & PCER & $c_h$ \\
  \hline
1000 & .01055   & 30.536 & .005 & .99832\\
& & &  .01 & .99548\\
& & & .02 & .98859\\
  \hline
\end{tabular}
\end{center}
\end{table}

By using our proposed two-stage procedure, we detect 102 out of 104 nonconforming vectors. In those 102 detected nonconforming vectors, we discover 3124, 4221, 5683 nonconforming features using PCER=0.005, 0.01, 0.02, respectively; and their respective average times to detect those nonconforming features are 61.886, 61.131, 60.274. By comparison, the one-stage LT procedure detects 89 nonconforming vectors. In those 89 detected nonconforming vectors, 2466 nonconforming features are detected and the average time to detect those nonconforming features is 64.116. From the above comparison, we can see that our two-stage procedure can provide both faster and more complete detection of nonconforming observations.

\section{Concluding Remarks}
In this paper, we introduce a general two-stage monitoring procedure for high-dimensional data streams. The proposed monitoring procedure can control both the IC-ARL and Type-I errors at the levels specified by users. Therefore, by using the proposed two-stage monitoring procedure, the users can choose how often they want to expect any false alarm  when the system is IC and how many false alarms they can tolerate when identifying OC data streams. This gives the users more freedom to design the desired monitoring procedure for high-dimensional data streams. Our simulation study shows that the proposed method has better performance comparing with the existing methods.

When monitoring high-dimensional data streams, depending on the purpose, two types of monitoring schemes are needed. The first type is for applications where any OC data stream indicates the same abnormality in the whole system and requires the same corrective action. As a result, those applications do not require the identification of OC data streams. For this monitoring purpose, a monitoring scheme which tracks a single global monitoring statistic such as those proposed in Tartakovsky et al. (2006), Mei (2010), Xie and Siegmund (2013), Zou et al. (2014) can achieve the goal.
The second type of monitoring schemes is for applications where different OC data streams indicate different problems in the system, and require different corrective actions. As a result, this type of monitoring schemes needs to identify which data streams are OC.
%In the literature, most of the existing methods for this type of monitoring schemes are to monitor %each of the data streams by one control chart and then adjust the control limit of each control chart %to satisfy the point-wise FDR or IC-ARL requirement.
The monitoring schemes proposed in in Grigg and Spiegelhalter (2008), Li and Tsung (2009, 2012), and Gandy and Lau (2013) belong to this category. In the past, the research on these two types of monitoring schemes has not overlapped. In this paper, the two-stage method we propose for the second type of monitoring schemes makes use of the development from the research on the first type of monitoring schemes. Therefore, our proposed two-stage method provides an important linkage between research on these two types of monitoring schemes.

In our proposed two-stage monitoring scheme, we choose the PCER as the Type-I error rate when identifying the OC data streams in our second stage. As mentioned in the paper, the main reason for this choice is that the global PCER can be easily controlled by controlling the point-wise PCER. In some applications, it might be acceptable to control only the point-wise Type-I error rate when identifying the OC data streams. If this is the case, the FDR might be a better choice than the PCER. However, to apply the existing FDR controlling procedures to our second stage essentially requires the calculation of the following $p$-values, denoted by $p^*_{i,t}$, which are not the marginal $p$-values, but are conditional on that $G_t>h$,
\[
p^*_{i,t}=P_{H_{0,i,t}}\Big(W_{i,t} > w_{i,t}|G_t>h \Big),
\]
where $w_{i,t}$ is the observed value of the test statistic $W_{i,t}$, the subscript ``$H_{0,i,t}$'' indicates that the probability is calculated under the null hypothesis $H_{0,i,t}$; that is, the $i$-th data stream is IC. The above conditional probability is very challenging to calculate. Therefore, the existing FDR controlling procedures cannot be directly applied to our second stage. In our future research, we will explore the estimation approach proposed by Storey (2002) to develop a suitable FDR controlling procedure for the second stage of our two-stage monitoring scheme.

\section*{Acknowledgements}
We thank the Editor, Professor Fugee Tsung, the Guest Editor, Professor Changliang Zou, and the 2 referees for their constructive comments and suggestions, which improved the quality of the paper greatly.

\section*{Appendix: Proof}
\begin{proposition}If $\ell_{1,t}(\cdot),..., \ell_{m,t}(\cdot)$ are all non-increasing functions, then we have, for $i=1,...,m_{0,t}$,
\[
P_{H_{1,t}}\Big(W_{i,t} > c_h|G_t>h \Big) \leq P_{H_{0,t}}\Big(W_{i,t} > c_h|G_t>h \Big).
\]
\end{proposition}
\begin{proof}
Since the proof will be the same for each time point $t$, for the sake of notational simplicity, we drop the subscript ``t'' in the proof below. Therefore, we need to prove that, for $i \leq m_{0}$,
\[
P_{H_{1}}\Big(W_{i} > c_h|G>h \Big) \leq P_{H_{0}}\Big(W_{i} > c_h|G>h \Big).
\]

Below we prove that
\begin{equation}
\label{eqn:conditP}
P_{H_{1}}\Big(W_{1} > c_h|G>h \Big) \leq P_{H_{0}}\Big(W_{1} > c_h|G>h \Big).
\end{equation}
For $i=2,..., m_{0}$, it can be proved similarly.

Based on the definition of conditional probabilities, it is easy to see that
\begin{align*}
& P_{H_{1}}\Big(W_{1} > c_h|G>h \Big)=\frac{P_{H_{1}}\Big(W_{1} > c_h \text{ and } G>h \Big)}{P_{H_{1}}\Big(G>h \Big)}\\
=&\frac{P_{H_{1}}\Big(W_{1} > c_h \text{ and } G>h \Big)}{P_{H_{1}}\Big(W_{1} > c_h \text{ and } G>h \Big)+P_{H_{1}}\Big(W_{1} < c_h \text{ and } G>h \Big)}\\
=&\left\{1+\frac{P_{H_{1}}\Big(W_{1} < c_h \text{ and } G>h \Big)}{P_{H_{1}}\Big(W_{1} > c_h \text{ and } G>h \Big)}\right\}^{-1}.
\end{align*}
Similarly we can write
\[
P_{H_{0}}\Big(W_{1} > c_h|G>h \Big) =\left\{1+\frac{P_{H_{0}}\Big(W_{1} < c_h \text{ and } G>h \Big)}{P_{H_{0}}\Big(W_{1} > c_h \text{ and } G>h \Big)}\right\}^{-1}.
\]
Therefore, to prove (\ref{eqn:conditP}), it suffices to show that,
\[
\frac{P_{H_{1}}\Big(W_{1} < c_h \text{ and } G>h \Big)}{P_{H_{1}}\Big(W_{1} > c_h \text{ and } G>h \Big)} \geq \frac{P_{H_{0}}\Big(W_{1} < c_h \text{ and } G>h \Big)}{P_{H_{0}}\Big(W_{1} > c_h \text{ and } G>h \Big)},
\]
or
\[
\frac{P_{H_{1}}\Big(W_{1} < c_h \text{ and } G>h \Big)}{P_{H_{0}}\Big(W_{1} < c_h \text{ and } G>h \Big)} \geq \frac{P_{H_{1}}\Big(W_{1} > c_h \text{ and } G>h \Big)}{P_{H_{0}}\Big(W_{1} > c_h \text{ and } G>h \Big)}.
\]

Denote the probability density function of $W_{i}$ under the null and alternative hypotheses by $f_{i}(\cdot)$ and $g_{i}(\cdot)$, respectively. Therefore,
\begin{align}
\label{eqn:less}
&\frac{P_{H_{1}}\Big(W_{1} < c_h \text{ and } G>h \Big)}{P_{H_{0}}\Big(W_{1} < c_h \text{ and } G>h \Big)}\nonumber\\
=&\frac{\int^{c_h}_{-\infty} \int_{u_2}^{\infty} \cdots\int_{u_m}^{\infty} f_{1}(w_1) \cdots f_{m_{0}}(w_{m_{0}}) g_{m_{0}+1}(w_{m_{0}+1}) \cdots g_{m}(w_m) dw_m \cdots dw_2 dw_1}{\int^{c_h}_{-\infty} \int_{u_2}^{\infty} \cdots\int_{u_m}^{\infty} f_{1}(w_1) \cdots f_{m_{0}}(w_{m_{0}}) f_{m_{0}+1}(w_{m_{0}+1}) \cdots f_{m}(w_m) dw_m \cdots dw_2 dw_1},
\end{align}
where $u_2$, ..., $u_m$ are the lower limits of integration with respect to $w_2$, ..., $w_m$, respectively, and $u_m$ is the function of $w_{1}$, ..., $w_{m-1}$, $u_{m-1}$ is the function of $w_1$, ..., $w_{m-2}$, and so on. Since larger $W_i$ leads to larger $G$, $u_2$, ..., $u_m$ are all non-increasing functions of $w_1$.

Similarly,
\begin{align}
\label{eqn:greater}
&\frac{P_{H_{1}}\Big(W_{1} > c_h \text{ and } G>h \Big)}{P_{H_{0}}\Big(W_{1} > c_h \text{ and } G>h \Big)}\nonumber\\
=&\frac{\int_{c_h}^{\infty} \int_{u_2}^{\infty} \cdots\int_{u_m}^{\infty} f_{1}(w_1) \cdots f_{m_{0}}(w_{m_{0}}) g_{m_{0}+1}(w_{m_{0}+1}) \cdots g_{m}(w_m) dw_m \cdots dw_2 dw_1}{\int_{c_h}^{\infty} \int_{u_2}^{\infty} \cdots\int_{u_m}^{\infty} f_{1}(w_1) \cdots f_{m_{0}}(w_{m_{0}}) f_{m_{0}+1}(w_{m_{0}+1}) \cdots f_{m}(w_m) dw_m \cdots dw_2 dw_1}.
\end{align}

Denote the $p$-value of $W_{i}$ by $p_{i}$ and its probability density function when the data stream is OC by $\ell_{i}(\cdot)$. Then (\ref{eqn:less}) and (\ref{eqn:greater}) can be calculated using their $p$-values. That is,
\[
\frac{P_{H_{1}}\Big(W_{1} < c_h \text{ and } G>h \Big)}{P_{H_{0}}\Big(W_{1} < c_h \text{ and } G>h \Big)}=\frac{\int_{v_1}^1 \int^{v_2}_0 \cdots\int^{v_m}_0 \ell_{m_{0}+1}(p_{m_{0}+1}) \cdots \ell_{m}(p_m) dp_m \cdots dp_2 dp_1}{\int_{v_1}^1 \int^{v_2}_0 \cdots\int^{v_m}_0 dp_m \cdots dp_2 dp_1},
\]
and
\[
\frac{P_{H_{1}}\Big(W_{1} > c_h \text{ and } G>h \Big)}{P_{H_{0}}\Big(W_{1} > c_h \text{ and } G>h \Big)}=\frac{\int_0^{v_1} \int^{v_2}_0 \cdots\int^{v_m}_0 \ell_{m_{0}+1}(p_{m_{0}+1}) \cdots \ell_{m}(p_m) dp_m \cdots dp_2 dp_1}{\int^{v_1}_0 \int^{v_2}_0 \cdots\int^{v_m}_0 dp_m \cdots dp_2 dp_1}.
\]
Here $v_1$ is the $p$-value $p_1$ when $W_1=c_h$; $v_2$, ..., $v_m$ are the lower limits of integration with respect to $p_2$, ..., $p_m$, respectively. Again $v_m$ is the function of $p_{1}$, ..., $p_{m-1}$,  $v_{m-1}$ is the function of $p_1$, ..., $p_{m-2}$, and so on, and $v_2$, ..., $v_m$ are all non-increasing functions of $p_1$.

%\textcolor{red}{
Define
\[
a(p_1) =\int^{v_2}_0 \cdots\int^{v_m}_0 dp_m \cdots dp_2,
\]
and
\[
b(p_1)=\frac{\int^{v_2}_0 \cdots\int^{v_m}_0 \ell_{m_{0}+1}(p_{m_{0}+1}) \cdots \ell_{m}(p_m) dp_m \cdots dp_2}{\int^{v_2}_0 \cdots\int^{v_m}_0 dp_m \cdots dp_2}.
\]
Since $\ell_{m_{0}+1}(\cdot),..., \ell_{m}(\cdot)$ are all non-increasing functions, and $v_2$, ..., $v_m$ are all non-increasing functions of $p_1$,   it is not difficult to see that $b(p_1)$ is a non-decreasing function of $p_1$.
Then
\begin{align*}
&\frac{P_{H_{1}}\Big(W_{1} < c_h \text{ and } G>h \Big)}{P_{H_{0}}\Big(W_{1} < c_h \text{ and } G>h \Big)}=\frac{\int_{v_1}^1 a(p_1)\cdot b(p_1) dp_1}{\int_{v_1}^1 a(p_1)dp_1} \geq \frac{\int_{v_1}^1 a(p_1)\cdot b(v_1) dp_1}{\int_{v_1}^1 a(p_1)dp_1}=b(v_1),\\
&\frac{P_{H_{1}}\Big(W_{1} > c_h \text{ and } G>h \Big)}{P_{H_{0}}\Big(W_{1} > c_h \text{ and } G>h \Big)}=\frac{\int_0^{v_1} a(p_1)\cdot b(p_1) dp_1}{\int_0^{v_1} a(p_1)dp_1} \leq \frac{\int_0^{v_1} a(p_1)\cdot b(v_1) dp_1}{\int_0^{v_1} a(p_1)dp_1}=b(v_1).
\end{align*} 
%}
Therefore, we have
\[
\frac{P_{H_{1}}\Big(W_{1} < c_h \text{ and } G>h \Big)}{P_{H_{0}}\Big(W_{1} < c_h \text{ and } G>h \Big)} \geq \frac{P_{H_{1}}\Big(W_{1} > c_h \text{ and } G>h \Big)}{P_{H_{0}}\Big(W_{1} > c_h \text{ and } G>h \Big)}.
\]
This completes the proof.
\end{proof}

%\begin{proof}[\textbf{Proof of Proposition \ref{thm:inequality}}]
%bla bla bla
%\end{proof}

\nocite*{}
%\bibliographystyle{apalike}
%\bibliography{CUSUM}

\end{document}